\documentclass[11pt]{article}
\usepackage{fullpage}
\usepackage{authblk}
\usepackage{amsthm,amsmath,amsfonts,amssymb}
\usepackage{xcolor}
\usepackage{bbm}
\usepackage{wrapfig}
\usepackage{enumerate}
\usepackage{hyperref}
\usepackage[capitalize]{cleveref}
\usepackage{tikz}
\usetikzlibrary{decorations.pathmorphing,patterns,decorations.markings,matrix,arrows,shadows}

\newtheorem{theorem}{Theorem}

\newtheorem{lemma}{Lemma}
\newtheorem{corollary}{Corollary}
\newtheorem{claim}{Claim}

\newcommand{\ignore}[1]{}
\newcommand{\boo}[1]{{\em\color{black} #1}}

\newif\ifnotesw\noteswtrue% T to show box & marginal notes; F suppresses.
   {\ifnotesw\marginpar[\hfill\(\top\)]{\(\top\)}\fi}%
   {\ifnotesw\marginpar[\hfill\(\bot\)]{\(\bot\)}\fi}

\newcommand{\mnote}[1]%
    {\ifnotesw\marginpar%
        [{\scriptsize\begin{minipage}[t]{\marginparwidth}
        \raggedleft#1%
                        \end{minipage}}]%
        {\scriptsize\begin{minipage}[t]{\marginparwidth}
        \raggedright#1%
                        \end{minipage}}%
    \fi}
\newcommand{\lip}[2]{\langle #1 , #2 \rangle}

\newcommand{\etal}{{\em et al.}~}

\newcommand{\OO}{\mathcal{O}}

\newcommand{\QQ}{\mathbb{Q}}
\newcommand{\ZZ}{\mathbb{Z}}
\newcommand{\RR}{\mathbb{R}}

\newcommand{\YY}{\mathcal{Y}}

\DeclareMathOperator{\Sp}{Sp}
\DeclareMathOperator{\Tr}{Tr}

\newcommand{\norm}[1]{\lVert #1 \rVert}

\newcommand{\trex}{{T}.{\em rex}~}

\title{Controlling quantum state transfer in rooted products\thanks{Corresponding author: ctamon@clarkson.edu}}
\author[1]{Addison Ballif}
\author[2]{Christino Tamon}
\author[2]{Gabriel Tucker}
\affil[1]{Department of Physics, Brigham Young University -- Idaho}
\affil[2]{Department of Computer Science, Clarkson University}
\date{\today}
\begin{document}
\maketitle

\begin{abstract}
Godsil and McKay (1978) showed that the rooted product is a powerful tool for constructing non-isomorphic cospectral pairs of graphs.
Despite lacking a convenient tensor product structure, we show that the rooted product is useful for 
constructing graphs with good quantum state transfer properties.
In particular, we prove a simple transference principle: if a graph $X$ has quantum state transfer and $Y$ is a {\em controllable} graph,
their rooted product $X^Y$ has quantum state transfer (inherited from $X$).
This complements a folklore property of Cartesian product which preserves perfect state transfer.
However, the rooted product is a significantly sparser graph and, more importantly, can be easily used to construct 
efficient high-fidelity state transfer even if $X$ has no quantum state transfer.
Our proof exploits the fact that a rooted product creates a large number of strongly cospectral pairs of vertices
and that its condition number can be controlled by its pendant subgraph.

\end{abstract}

%%%%%%%%%%%%%%%%%%%%%%%%%%%%%%%%%%%%%%%%%%%%%%%%%%%%%%%%%%%%%%%%%%%%%%%%%%%%%%%%%%%%%%%%%%%%%%%%%%%%%
\section{Introduction} 

In a seminal work, Bose \cite{b03} proposed the study of quantum state transfer via continuous-time quantum walk on graphs.
Given a graph $X$ with adjacency matrix $A(X)$ and two chosen vertices $a,b \in V(X)$, the goal of quantum state transfer
is to evolve the unit vector $e_a$ to another unit vector $e_b$ using a quantum walk given by $e^{-itA(X)}$.
We measure the success of this protocol using its transfer fidelity 
\[
	f(t) = |\lip{e_b}{e^{-it A(X)}e_a}|.
\]
If the quantum walk succeeds with unit fidelity, that is, $f(\tau)=1$, for some time $\tau$,
we say {\em perfect} state transfer (PST) occurs between $a$ and $b$ at time $\tau$. 
An important consequence is that we can transmit quantum information (in a lossless manner) 
between the two sites in a quantum network governed by $X$.
Since its inception, quantum state transfer has gained considerable interest 
due to its importance in quantum information (see Kay \cite{k11} and Godsil \cite{g12dm} for relevant surveys).

But, perfect state transfer (PST) is a rare phenomenon on unweighted graphs.
Godsil proved that for any $k$, the number of graphs of maximum degree $k$ that has perfect state transfer is finite
(see \cite{g12b}).
This motivates a relaxation called {\em pretty good} state transfer (PGST) (see \cite{gkss12,vz12})
where we merely demand that fidelity can be made arbitrarily close to $1$, that is, 
for any $\epsilon > 0$, there {\em exists} a time $\tau$ for which $f(\tau) \ge 1-\epsilon$.
A second relaxation is to allow some minimal amount of weights used in our graphs for quantum state transfer. 
This is to minimize the cost of engineering or preparing the different interaction strengths (edge weights)
between sites (vertices) in the associated quantum spin network (or underlying graph); see \cite{k11}.
So, the main problem is to design families of graphs with good fidelity for quantum state transfer
while utilizing few weighted edges.

In this work, we return to a known graph operator called the {\em rooted product} introduced by Godsil and McKay \cite{gm78}.
If $X$ is a base graph and $Y$ is a rooted graph with a distinguished root vertex $r$, their rooted product $X^{Y}$
is obtained by attaching (or 'rooting') a separate copy of $Y$ to each vertex $a$ of $X$ 
(where we identify $a$ with $r$ in $Y$).
In algebraic graph theory, Godsil and McKay employed the rooted product as a powerful tool 
for constructing non-isomorphic families of graphs which are cospectral.

\begin{figure}[h]
\begin{center}
\begin{tikzpicture}[
% T.rex
    main node/.style={circle,draw,font=\bfseries}, main edge/.style={-,>=stealth'},
    scale=0.5,
    stone/.style={},
    black-stone/.style={black!80},
    black-highlight/.style={outer color=black!80, inner color=black!30},
    black-number/.style={white},
    white-stone/.style={white!70!black},
    white-highlight/.style={outer color=white!70!black, inner color=white},
    white-number/.style={black}]
\tikzset{every loop/.style={thick, min distance=17mm, in=45, out=135}}

% to show particle, uncomment the next line
%\gustone[0]{black}{-3}{1.25}

% ellipse
\draw[fill={gray!20}, drop shadow]
    (-10.0,3.5) ellipse (1.25cm and 3.0cm)
    (-6.0,3.5) ellipse (1.25cm and 3.15cm)
    (-2.0,3.5) ellipse (1.25cm and 3.15cm)
    (+2.0,3.5) ellipse (1.25cm and 3.15cm)
    (+6.0,3.5) ellipse (1.25cm and 3.15cm)
    (+10.0,3.5) ellipse (1.25cm and 3.0cm);

\tikzstyle{every node}=[draw, shape=circle, fill={gray!20}, drop shadow];
\path (-10.0,1.5) node [scale=0.8] (aa1) {$a$};
\path (+10.0,1.5) node [scale=0.8] (aa6) {$b$};
\path (-10.0,5.5) node [scale=0.8] (b1) {$y_a$};
\path (-6.0,1.0) node [scale=0.8] (aa2) {};
\path (-2.0,1.0) node [scale=0.8] (aa3) {};
\path (+2.0,1.0) node [scale=0.8] (aa4) {};
\path (+6.0,1.0) node [scale=0.8] (aa5) {};
\path (+10.0,5.5) node [scale=0.8] (b6) {$y_b$};

% base path
\draw[double, double distance=1pt, color=gray]
    (aa1) -- (aa2)
    (aa2) -- (aa3)
    (aa3) -- (aa4)
    (aa4) -- (aa5)
    (aa5) -- (aa6);

\tikzstyle{every node}=[];
\path (-10.0,3.5) node (q1) {$Y$};
\path (-6.0,3.5) node (q1) {$Y$};
\path (-2.0,3.5) node (q1) {$Y$};
\path (+2.0,3.5) node (q1) {$Y$};
\path (+6.0,3.5) node (q1) {$Y$};
\path (+10.0,3.5) node (q1) {$Y$};

\end{tikzpicture}
\caption{The rooted product $X^Y$ of a base graph $X$ with a rooted graph $Y$, where $Y$ is
controllable at its root. As an example, here $X$ is a path connecting $a$ and $b$. 
If there is state transfer between $a$ and $b$ in $X$, 
then there is state transfer between $y_a$ and $y_b$ in $X^Y$, for all $y$ in $Y$.
}
\label{fig:rooted}
\end{center}
\end{figure}
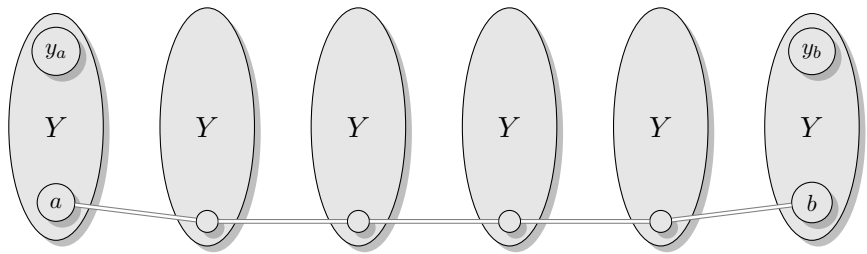

For quantum state transfer, we prove a simple transference principle on rooted products:
if $X$ has state transfer between two of its vertices $a$ and $b$, and $Y$ is a controllable graph,
then $X^Y$ (up to a scaling of $X$) has pretty good state transfer between {\em every} pair 
of vertices $y_a=(y,a)$ and $y_b=(y,b)$; see \cref{fig:rooted}.
Here, $y_a=(y,a)$ denotes the vertex $y$ of $Y$ in the pendant subgraph rooted at vertex $a$ of $X$.
Thus, we may also view the rooted product as a {\em closure} operator for quantum state transfer.
This is similar to a folklore property of the Cartesian product which preserves perfect state transfer 
in its constituent graphs (a famous example being the hypercubes \cite{cddekl05}).
However, the rooted product creates a significantly sparser family of graphs compared to the Cartesian product.

A useful property of rooted product is that it creates a large number of pairs of
strongly cospectral vertices. This facilitates pretty good state transfer between many pairs of vertices.
An earlier construction that creates a large number of strongly cospectral vertices using Cartesian products 
was given by Chan and Sin \cite{cs24} (see also Pal and Bhattacharjya \cite{pb17}, Example 4.1).
As mentioned above, a notable feature of the rooted product is its ability to create {\em sparse} graphs with such property.
For example, consider the comb graph, which is a rooted product of $P_m$ and $P_n$, in contrast to the
denser grid $P_m \Box P_n$.

Using the transference principle, we derive the following corollaries.
We show pretty good state transfer in a $1$-sum of paths with controllable graphs which requires only one weighted edge.
This is similar to a result of Godsil \etal \cite{ggklm20} where a single weighted edge helps create quantum state transfer
in strongly regular graphs.
Also, by taking the rooted product of a graph with universal perfect state transfer 
(which has perfect state transfer between every pair of vertices) with a controllable graph, 
we obtain a family of graphs with multiple pretty good state transfer 
(within each `layer' of the pendant controllable subgraph).
This extends a result in Acuaviva \etal \cite{aceghtwz25} where the controllable graph is a path.

In the second part, which is our main contribution of this work, we address two known
problematic issues with pretty good state transfer: 
{\em non-deterministic} arrival time and the use of {\em transcendental} weights.
Even though pretty good state transfer allows us to obtain fidelity that is arbitrarily close to one,
the required time is nondeterministic as it is purely existential and non-constructive.
So, we focus on {\em high-fidelity} state transfer which requires an explicit asymptotic 
dependence between time and the fidelity gap.
In particular, if the fidelity is at least $1-\epsilon$, we demand that the time is a polynomial function
of $1/\epsilon$ and the size of the graph.

We show that the rooted product $X^Y$ exhibits high-fidelity state transfer between a large number of vertices 
in the subgraphs of $Y$ by attaching sufficiently weak pendant edges.
Here, we apply the \trex method (see Kay and Tamon \cite{kt}) which requires that the two vertices involved 
in state transfer are cospectral and that the condition number of the graph is sufficiently bounded.
In fact, we show that how to use strong cospectrality between the vertices (which is guaranteed by the
rooted product) to obtain bounds on the running time of the \trex scheme.
Moreover, we show that the condition number of the rooted product $X^Y$ can be controlled using 
the condition number of the pendant graph $Y$ (and its vertex deleted subgraph).
Finally, this result does not require scaling the base graph $X$ with a transcendental weight
(thereby, minimizing the use of weighted edges) as required in our results on pretty good state transfer.

Our contribution in this work is to show that, for efficient and constructive quantum state transfer, 
the rooted product is a powerful tool due to the guaranteed {\em strong} cospectrality and 
the {\em controllability} of its {\em condition} number.

%%%%%%%%%%%%%%%%%%%%%%%%%%%%%%%%%%%%%%%%%%%%%%%%%%%%%%%%%%%%%%%%%%%%%%%%%%%%%%%%%%%%%%%%%%%%%%%%%%%%%
\section{Preliminaries} 

We describe terminology and notation used throughout this work.
For a positive integer $n$, let $[n]$ denote the set $\{1,\ldots,n\}$.
We adopt standard notation from graph theory (see Godsil and Royle \cite{gr}).
The spectrum of a graph $X$, denoted $\Sp(X)$, is the set of eigenvalues of its adjacency matrix $A(X)$,
which is also the set of zeros of its characteristic polynomial $\phi(X,t) = \det(tI - A(X))$.

For a graph $X$ and scalar $\alpha \in \RR$, let $X_\alpha$ denote the graph whose adjacency matrix is $\alpha A(X)$.

A rooted graph is a pair $(Y,r)$ where $Y$ is a graph and $r \in V(Y)$ is a distinguished vertex called the root.
Given a graph $X$ with vertex set $[n]$ and a collection of rooted graphs $\YY = \{(Y_k,r_k) : k=1,\ldots,n\}$,
the {\em rooted product} of $X$ with $\YY$ (see Godsil and McKay \cite{gm78}), denoted $X^\YY$, 
is the graph obtained from the disjoint union $\bigcup_{k=1}^{n} Y_k$ by adding the edges $(r_j,r_k)$ 
whenever $(j,k) \in E(X)$.
If $(Y_k,r_k)$ are identical to $(Y,r)$ for all $k$, we write $X^{(Y,r)}$ or simply $X^Y$ if $r$ is clear from context. 
In this case, the adjacency matrix of $X^{(Y,r)}$ is given by
\begin{equation} \label{eqn:rooted-adj}
	A(X^{(Y,r)}) = A(Y) \otimes I_n + e_r e_r^T \otimes A(X)
\end{equation}
where $n$ is the order of $X$.
Based on the decomposition in \cref{eqn:rooted-adj}, we may index each vertex of $X^Y$ as $(y,x)$,
for $y \in Y$ and $x \in X$.

For a graph $X$ of order $n$ and a vector $z$, the {\em walk matrix} $W_X(z)$ is a matrix whose $k$-th column is $A(X)^k z$; 
that is:
\[
	W_X(z) = 
	\begin{pmatrix}
	\vdots 	& \vdots	& \vdots	& \vdots & \vdots		\\
	z 		& A(X)z 	& A(X)^2 z 	& \ldots & A(X)^{n-1}z 	\\
	\vdots 	& \vdots	& \vdots	& \vdots & \vdots		
	\end{pmatrix}.
\]
In what follows, we use $e_a$ to denote the unit vector that is $1$ at $a$ and $0$ elsewhere.
A vertex $a \in V(X)$ is controllable in $X$ if $W_X(e_{a})$ has full rank.
The pair $(X,a)$ is called {\em controllable} if $a$ is controllable in $X$.
Controllable graphs were studied in \cite{crsy11,g12}.

\bigskip
Next, we describe relevant terminology for quantum walks and state transfer.
For a graph $X$, let $A(X) = \sum_{s=1}^{d} \lambda_s E_s$ be the spectral decomposition of $A(X)$.
The continuous-time quantum walk on $X$ is defined by the time-varying unitary matrix
\[
	e^{-itA(X)} = \sum_{s=1}^{d} e^{-it\lambda_s}E_s
\]
as time $t$ ranges in $(-\infty,+\infty)$.

Two vertices $a,b \in V(X)$ are called {\em cospectral} if $e_a^T E_s e_a = e_b^T E_s e_b$ for all $s \in [d]$.
The vertices are {\em strongly cospectral} if for each eigenvalue $\lambda_s$, we have
\begin{equation} \label{eqn:quarrel}
E_s e_a = e^{iq_{a,b}(\lambda_s)} E_s e_b,
\end{equation}
for some phase factor $q_{a,b}(\lambda_s) \in [0,2\pi)$ (the so-called {\em quarrel} of \cite{gl20}). 

We say {\em pretty good state transfer} occurs between $a,b \in V(X)$ if 
for any $\epsilon \ge 0$, there is a time $\tau > 0$ so that 
\begin{equation} \label{eqn:def_pgst}
\lip{e_b}{e^{-i\tau A(X)}e_a} = \rho e^{i\varphi}
\end{equation}
where $\rho \in \RR$ is a real number which satisfies $|\rho| \ge 1-\epsilon$ and $\varphi \in [0,2\pi)$.
The state transfer is {\em perfect} if there is a time $\tau$ so that $|\rho|=1$ holds.
By applying spectral decomposition and strong cospectrality in \cref{eqn:def_pgst}, we get that
\begin{equation} \label{eqn:triple-one}
\rho 
	\ = \ \sum_{s=1}^{d} e^{-i(\lambda_s \tau + \varphi)} \lip{e_b}{E_s e_a} 
    \ = \ \sum_{s=1}^{d} e^{i(q_{a,b}(\lambda_s) - \varphi - \lambda_s \tau)} \lip{e_b}{E_s e_b}.
\end{equation}
Next, applying triangle inequality and resolution of identity to \cref{eqn:triple-one}, we get
\begin{equation}
1-\varepsilon \ \le \ |\rho| \ \le \ \sum_{s=1}^{d} |\lip{e_b}{E_s e_b}| \ = \ \sum_{s=1}^{d} \lip{e_b}{E_s e_b} \ = \ 1.
\end{equation}
Thus, we require the time $\tau$ to satisfy
\begin{equation}
    |\lambda_s\tau - q_{a,b}(\lambda_s) + \varphi| < \epsilon \pmod{2\pi},
\end{equation}
for each eigenvalue $\lambda_s$ of $A(X)$.
Hence, the following theorem due to Kronecker will be useful.

\begin{theorem} \label{thm:kronecker} (see Levitan and Zhikov \cite{lz}) \\
Let $\lambda_1,\ldots,\lambda_d$ and $\theta_1,\ldots,\theta_d$ be arbitrary real numbers.
For an arbitrarily small $\epsilon$, the system of inequalities
\begin{equation}
    |\lambda_s\tau - \theta_s| < \epsilon \pmod{2\pi}, 
    \ \ \
    s=1,\ldots,d
\end{equation}
admits a solution for $\tau$ if and only if 
\begin{equation}
(\forall \ell_1,\ldots,\ell_d \in \ZZ)
    \left[ \ell_1\lambda_1 + \ldots + \ell_d\lambda_d = 0
	\ \implies \
    \ell_1\theta_1 + \ldots + \ell_d\theta_d \equiv 0\pmod{2\pi}\right].
\end{equation}
\end{theorem}

We state a characterization theorem for pretty good state transfer.

\begin{theorem} \label{thm:pgst} (see Theorem 5.2 in Acuaviva \etal \cite{aceghtwz25}) \\
Let $X$ be a graph with eigenvalues $\lambda_1,\ldots,\lambda_d$ in the eigenvalue support of $a \in V(X)$.
Then, $X$ has pretty good state transfer from $a$ to $b$ if and only if the following hold:
\begin{enumerate}[(i)]
\item The vertices $a$ and $b$ are strongly cospectral with quarrels $q_{a,b}(\lambda_j)$, $j=1,\ldots,d$.
\item There exists $\delta \in \RR$ so that for all $\ell_1,\ldots,\ell_d \in \ZZ$:
	\begin{equation} \label{eqn:pgst-imply}
	\sum_{s=1}^{d} \ell_s\lambda_s = 0
	\ \implies \
	\sum_{s=1}^{d} \ell_s(q_{a,b}(\lambda_s) + \delta) \equiv 0\pmod{2\pi}.
	\end{equation}
\end{enumerate}
\end{theorem}

Recall that the parameter $\delta$ is included in the second condition of \cref{thm:pgst}
to account for the phase factor in the fidelity $\lip{e_b}{e^{-itA(X)}e_a}$.

%%%%%%%%%%%%%%%%%%%%%%%%%%%%%%%%%%%%%%%%%%%%%%%%%%%%%%%%%%%%%%%%%%%%%%%%%%%%%%%%%%%%%%%%%%%%%%%%%%%%%
\section{Transference}

In this section, we prove a transference principle for quantum state transfer in rooted products.

\begin{theorem} \label{thm:closure}
Let $X$ be a graph with pretty good state transfer between $a,b \in V(X)$ and 
let $Y$ be a graph controllable at vertex $r$.
Then, there are infinitely many real numbers $\alpha$ so that
the rooted product $X_\alpha^{(Y,r)}$ has pretty good state transfer between $(y,a)$ and $(y,b)$, 
for all $y \in V(Y)$.
\end{theorem}

\par\noindent{\em Remark}. It is convenient here to think of $\alpha$ as transcendental 
although this assumption can be removed in most subsequent corollaries.

\medskip

\begin{proof}
We will present the proof of \cref{thm:closure} in three parts. 
First, we locate the eigenvalues of the rooted product using a determinant identity from \cite{gm78}.
Second, we show a tensor factorization of the eigenprojectors of $X_\alpha^Y$,
Finally, we apply \cref{thm:kronecker} coupled with the trace method from \cite{kly17}.
Our proof generalizes a result in \cite{aceghtwz25} (see Section 5(b)).

\subsection{Eigenvalues}

Let $F_0$ be the smallest field extension of $\QQ$ which contains all eigenvalues of $X$ and $Y$.
Note that all eigenvalues of $X_\alpha$ lie in $F=F_0(\alpha)$.

For a real scalar $\lambda$, let $\phi_\lambda(t)$ be the characteristic polynomial of the graph 
obtained from $Y$ by adding a self-loop of weight $\lambda$ to its root vertex $r$. 
Note that
\begin{equation} \label{eqn:char-poly}
	\phi_\lambda(t) = \phi(Y,t) - \lambda\phi(Y\setminus r,t).
\end{equation}
The following lemma shows that the characteristic polynomial of the rooted product is a product
of the characteristic polynomials of the pendant subgraph $Y$ after a rank-one perturbation induced
by the individual eigenvalues of $X$.

\begin{lemma} \label{lemma:gm-decomposition}
The characteristic polynomial of the rooted product $X_\alpha^Y$ is given by
\[
	\phi(X_\alpha^Y,t) = \prod_{\lambda \in \Sp(X_\alpha)} \phi_\lambda(t)
\]

\begin{proof}
We apply the following determinantal identity (due to Godsil and McKay \cite{gm78}):
\begin{equation} \label{eqn:godsil-mckay}
	\phi(X_\alpha^Y,t) 
	= \det(\phi(Y,t)I - \phi(Y\setminus r,t)A(X_\alpha)).
\end{equation}
After unrolling the determinant and applying \cref{eqn:char-poly}, the claim follows.
\end{proof}
\end{lemma}

Let us denote the zeros of $\phi_\lambda(t)$ as $\theta_j(\lambda)$ as $j=1,\ldots,d$.

\begin{lemma} \label{lemma:irreducible}
For each $\lambda \in \Sp(X_\alpha)$, $\phi_\lambda(t)$ is irreducible over $F=F_0(\alpha)$.

\begin{proof}
Suppose $\phi_\lambda(t)$ is reducible over $F[t]$; so, $\phi_\lambda(t) = p(t)q(t)$ for some $p(t),q(t) \in F[t]$.
Each eigenvalue $\lambda$ of $X_\alpha$ is of the form $\lambda = \alpha\theta$ where $\theta \in \Sp(X)$. 
Note $\theta \in F_0$.
Thus,
\[
	\phi_\lambda(t) = \phi(Y,t) - \alpha\theta \phi(Y\setminus r,t),
\]
which shows $\phi_\lambda(t)$ is linear in $\alpha$. 
Hence, one of the factors of $\phi_\lambda(t)$, say $p(t)$, must be an element of $F_0[t]$
(here, we used the fact that $\alpha$ is \boo{transcendental}).
It follows that $p(t)$ divides both $\phi(Y,t)$ and $\phi(Y\setminus r,t)$.
But, $\phi(Y,t)$ and $\phi(Y\setminus r,t)$ are coprime over $F_0[t]$ 
since $(Y,r)$ is \boo{controllable} (see Lemma 7.2 in Godsil \cite{g12}).
This contradiction proves the claim.
\end{proof}
\end{lemma}

\subsection{Eigenprojectors}

The following lemma is a generalization of a similar observation in \cite{aceghtwz25} which was
specific for weighted paths (or Jacobi matrices). It shows that the eigenprojectors of the rooted
products admits a convenient tensor products decomposition (not unlike the eigenprojectors of the
Cartesian product). Moreover, it also shows that (strong) cospectrality is an inherited property
in a rooted graph product.

\begin{lemma} \label{lemma:chiral} 
For any graph $X$ and any rooted graph $(Y,r)$, consider the rooted product $X^{(Y,r)}$.
For each eigenvalue $\lambda$ of $X$, let $\theta_j(\lambda)$, $j=1,\ldots,d$, 
be the eigenvalues of $J_\lambda = A(Y) + \lambda e_r e_r^T$.
Then, $\theta_j(\lambda)$ is an eigenvalue of $X^{(Y,r)}$ with eigenprojector
\begin{equation} \label{eqn:tensor-decomp}
	E_{\theta_j(\lambda)}(X^{(Y,r)}) = E_{\theta_j(\lambda)}(J_\lambda) \otimes E_\lambda(X).
\end{equation}
Moreover, for $a,b \in V(X)$ and $y \in V(Y)$, the following ({\em inherited quarrel}) property holds:
\begin{equation} \label{eqn:shared-quarrel}
	q_{(y,a),(y,b)}(\theta_j(\lambda)) = q_{a,b}(\lambda),
	\ \ \
	j=1,\ldots,d.
\end{equation}

\begin{proof}
Note that
\[
	A(X^{(Y,r)}) = A(Y) \otimes I + e_r e_r^T \otimes A(X).
\]
Suppose $x$ is a $\lambda$-eigenvector of $A(X)$ and
$y$ is a $\theta_j(\lambda)$-eigenvector of $A(Y) + \lambda e_r e_r^T$. Then,
\begin{align*}
A(X^{(Y,r)}) y \otimes x 
	&= A(Y)y \otimes x + e_r e_r^T y \otimes A(X)x \\
	&= A(Y)y \otimes x + \lambda e_r e_r^T y \otimes x \\
	&= (A(Y) + \lambda e_r e_r^T) y \otimes x \\
	&= \theta_j(\lambda) y \otimes x.
\end{align*}
This proves \cref{eqn:tensor-decomp}. Finally, \cref{eqn:shared-quarrel} follows from this.
\end{proof}
\end{lemma}

Note that \cref{lemma:chiral} holds independently of the choice of $\alpha$ (transcendental or otherwise).

\subsection{Kronecker approximation}

We now show that the second condition in \cref{thm:pgst} holds.
For each eigenvalue $\lambda$ of $X_\alpha$, suppose there are integers 
$\ell_j(\lambda) \in \ZZ$, $j=1,\ldots,d$, so that
\begin{equation} \label{eqn:kronecker-premise}
	\sum_{j \in [d]} \sum_{\lambda \in \Sp(X_\alpha)} \ell_j(\lambda)\theta_j(\lambda) = 0.
\end{equation}
We must show that 
\begin{equation} \label{eqn:kronecker-goal}
	\sum_{j \in [d]} \sum_{\lambda \in \Sp(X_\alpha)} \ell_j(\lambda) q_j(\lambda) \equiv 0\pmod{2\pi}
\end{equation}
where we had used $q_j(\lambda)$ to denote the quarrel
\begin{equation}
q_j(\lambda) := q_{(y,a),(y,b)}(\theta_j(\lambda)).
\end{equation}
Here, we fix and assume that $a$ and $b$ are with pretty good state transfer in $X$.

Let $F_\lambda$ be the splitting field of $\phi_\lambda(t)$ over $F$.
Let $M$ be the smallest field extension of $F$ which contains all $F_\lambda$, for distinct $\lambda \in \Sp(X_\alpha)$.
Thus, $M$ is the splitting field of $\phi(X_\alpha^Y,t)$ (see \cref{lemma:gm-decomposition}).
We now employ the trace method (due to Kempton, Lippner, and Yau \cite{kly17}).

Applying the field trace (see Roman \cite{r95}, Section 7.1) to both sides of \cref{eqn:kronecker-premise}, we obtain
\begin{equation} \label{eqn:premise}
	\sum_{j \in [r]} \sum_{\lambda \in \Sp(X_\alpha)} \ell_j(\lambda) \Tr_{M/F}(\theta_j(\lambda)) = 0.
\end{equation}
Note that
\[
	\Tr_{M/F}(\theta_j(\lambda))
	= [M:F_\lambda] \Tr_{F_\lambda/F}(\theta_j(\lambda))
\]
and
\[
	\Tr_{F_\lambda/F}(\theta_j(\lambda))
	= \frac{[F_\lambda:F]}{\deg(\theta_j(\lambda))}[t^{m-1}]m_{\theta_j(\lambda)}(t)
	= \frac{[F_\lambda:F]}{|V(Y)|} \lambda
\]
where $\deg(\omega)$ denotes the degree of the minimal polynomial $m_\omega(t)$ of the field element $\omega$.
We have used the fact that the minimal polynomial of $\theta_j(\lambda)$ is equal to $\phi_\lambda(t)$
as the latter is irreducible over $F$ (see \cref{lemma:irreducible}).
Therefore,
\[
	\Tr_{M/F}(\theta_j(\lambda))
	= \frac{[M:F]}{|V(Y)|} \lambda
\]
by the tower law $[M:F] = [M:F_\lambda][F_\lambda:F]$.
Returning to \cref{eqn:premise}, we obtain
\begin{equation} \label{eqn:pgst-X}
	\sum_{\lambda \in \Sp(X_\alpha)} L_\lambda \lambda = 0,
	\hspace{.5in} \mbox{ where $L_\lambda = \sum_{j \in [d]} \ell_j(\lambda)$.}
\end{equation}
Now, we use the assumption that $X_\alpha$ has pretty good state transfer between $a,b \in V(X)$.
Since \cref{eqn:pgst-X} holds, we have
\[
	\sum_{\lambda \in \Sp(X_\alpha)} L_\lambda (q_{a,b}(\lambda) + \varphi) \equiv 0\pmod{2\pi}
\]
for some $\varphi \in [0, 2\pi)$.
By \cref{eqn:shared-quarrel}, $q_j(\lambda) = q_{a,b}(\lambda)$ for all $j \in [d]$. 
Thus, we have
\[
	\sum_{\lambda \in \Sp(X_\alpha)} \sum_{j \in [d]} \ell_j(\lambda) (q_{a,b}(\lambda) + \varphi) 
	= \sum_{\lambda \in \Sp(X_\alpha)} \sum_{j \in [d]} \ell_j(\lambda) (q_j(\lambda) + \varphi)
	\equiv 0\pmod{2\pi}
\]
which proves \cref{eqn:kronecker-goal}.
This shows that the second condition of \cref{thm:pgst} is satisfied.

Recall we had shown that the first condition of \cref{thm:pgst} is satisfied in \cref{lemma:chiral};
that is, $(y,a)$ and $(y,b)$ in $X_\alpha^Y$ are strongly cospectral with a shared quarrel.
Thus, by \cref{thm:pgst}, $X_\alpha^Y$ admits pretty good state transfer between the vertices
$(y,a)$ and $(y,b)$ for all $y \in V(Y)$.

This completes the proof (of \cref{thm:closure}).
\end{proof}

%%%%%%%%%%%%%%%%%%%%%%%%%%%%%%%%%%%%%%%%%%%%%%%%%%%%%%%%%%%%%%%%%%%%%%%%%%%%%%%%%%%%%%%%%%%%%%%%%%%%%
\section{Power of edge weights} \label{sec:power}

\begin{lemma} \label{lem:p2_controllable}
For any controllable graph $(H,r)$, the rooted product $(\alpha P_2)^{(H,r)}$ 
has pretty good state transfer between corresponding vertices of $H$, for infinitely many $\alpha \in \RR$.

\begin{proof}
Take $P_2$ as the base graph and apply \cref{thm:closure}.
\end{proof}
\end{lemma}

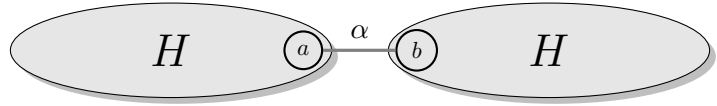
\begin{figure}[h]
\begin{center}
\begin{tikzpicture}[
% T.rex
    main node/.style={circle,draw,font=\bfseries}, main edge/.style={-,>=stealth'},
    scale=0.5,
    stone/.style={},
    black-stone/.style={black!80},
    black-highlight/.style={outer color=black!80, inner color=black!30},
    black-number/.style={white},
    white-stone/.style={white!70!black},
    white-highlight/.style={outer color=white!70!black, inner color=white},
    white-number/.style={black}]
\tikzset{every loop/.style={thick, min distance=17mm, in=45, out=135}}

% to show particle, uncomment the next line
%\gustone[0]{black}{-3}{1.25}

% ellipse
\draw[fill={gray!20}, drop shadow]
    (-5.0,0.0) ellipse (4.25cm and 1.25cm)
    (+5.0,0.0) ellipse (4.25cm and 1.25cm);

\tikzstyle{every node}=[draw, thick, shape=circle, fill={gray!20}];
\path (-1.5,0.0) node [scale=0.8] (a1) {$a$};
\path (+1.5,0.0) node [scale=0.8] (b1) {$b$};

\tikzstyle{every node}=[];
\path (-5.0,0.0) node (qleft) {\mbox{\LARGE $H$}};
\path (+5.0,0.0) node (qright) {\mbox{\LARGE $H$}};

% middle/bridge edge
\draw[very thick, color=gray]
    (a1) -- (b1);

\tikzstyle{every node}=[];
\node at (0.0,0.5) {$\alpha$};

\end{tikzpicture}
\caption{The rooted product of $P_2$ with a controllable graph $H$:
pretty good state transfer occurs symmetrically across the weighted edge $(a,b)$
for infinitely many $\alpha$ (even as $\alpha \rightarrow 0$).
This complements a result of Kempton-Lippner-Yau \cite{kly17} where 
$\alpha=1$ and $H$ is a path with a weighted loop on a pendant vertex.
}
\label{fig:p2_rooted}
\end{center}
\end{figure}

A special case of \cref{lem:p2_controllable} is the path $P_{2n}$ where the middle edge is weighted (see \cref{fig:p2_rooted}). 
This complements the result of Kempton \etal \cite{kly17} where, instead of weighting the middle edge, 
they added weighted loops to the two pendant vertices. So, we traded two weighted loops for a single weighted edge.
It is curious that \cref{lem:p2_controllable} holds even for an infinite sequence of $(\alpha_m)_{m=1}^{\infty}$ 
which tends to zero as $m \rightarrow \infty$ (in the limit, the two controllable graphs 
are disconnected from each other).

Let $X_1$ and $X_2$ be two graphs. The $1$-sum of $X_1$ and $X_2$ is a graph $X$
where $V(X) = V(X_1) \cup V(X_2)$ with $|V(X_1) \cap V(X_2)| = 1$ and $E(X) = E(X_1) \cup E(X_2)$.
If $r \in V(X_1) \cap V(X_2)$, then we say $X$ is a $1$-sum of $X_1$ and $X_2$ joined at vertex $r$.

\begin{figure}[h]
\begin{center}
\begin{tikzpicture}[
    main node/.style={circle,draw,font=\bfseries}, main edge/.style={-,>=stealth'},
    scale=0.5,
    stone/.style={},
    black-stone/.style={black!80},
    black-highlight/.style={outer color=black!80, inner color=black!30},
    black-number/.style={white},
    white-stone/.style={white!70!black},
    white-highlight/.style={outer color=white!70!black, inner color=white},
    white-number/.style={black}]
\tikzset{every loop/.style={thick, min distance=17mm, in=45, out=135}}

% to show particle, uncomment the next line
%\gustone[0]{black}{-3}{1.25}

% ellipse
\draw[fill={gray!20}, drop shadow]
    (+10.0,0.0) ellipse (3.1cm and 1.75cm);

\tikzstyle{every node}=[draw, thick, shape=circle, fill={gray!20}];
\path (+0.5,0.0) node [scale=0.7] (b1) {$1$};
\path (+7.5,0.0) node [scale=0.7] (bb1) {$m$};

\tikzstyle{every node}=[];
\path (+4.0,1.0) node (qright) {\mbox{$P_m$}};
\path (+10.0,0.0) node (qqright) {\mbox{\LARGE $H$}};

\draw[very thick, color=gray]
    (b1) -- (bb1);

\end{tikzpicture}
\caption{A $1$-sum $X$ of a path $P_m$ and a graph $H$ joined at vertex $m$. If $H$ is controllable at $m$, 
then $X$ is controllable at $1$. 
}
\label{fig:lollipop}
\end{center}
\end{figure}
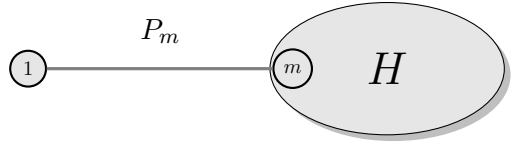

Next, we show that a $1$-sum of a path with an arbitrary controllable graph is controllable provided they
are joined at the common controllable vertex (see \cref{fig:lollipop}).

\begin{lemma}
Let $P_m$ be a path on vertices $\{1,\ldots,m\}$ and let $H$ be a graph on vertices $\{m,\ldots,m+n-1\}$.
Suppose that $H$ is controllable at vertex $m$. Let $X$ be the $1$-sum of $P_m$ and $H$ joined at vertex $m$.
Then, $X$ is controllable at vertex $1$.

\begin{proof}
Let $A_1 = A(P_m) \oplus O_{n-1}$ be the adjacency matrix of the induced subgraph $P_m$ of $X$ and
let $A_2 = O_{m-1} \oplus A(H)$ be the adjacency matrix of the induced subgraph $H$ of $X$.
Here, both $A_1$ and $A_2$ are of order $m+n-1$.
Let $e_1,\ldots,e_{m+n-1}$ be the standard basis vectors for $X$.

Let $z_k = A_X^k e_1$ for $k=0,\ldots,m+n-2$. To show that $X$ is controllable at $1$, we must prove that 
the set $\{z_0,\ldots,z_k\}$ is linearly independent for all $k=0,\ldots,m+n-2$.
We divide this into two cases.
For $k=0,\ldots,m-1$, the claim holds since $z_k = A_1^k e_1$ as $P_m$ is controllable at vertex $1$.
For $k=m,\ldots,m+n-2$, we first show that $z_k$ can be written in the following form:
\begin{equation} \label{eqn:induct}
	z_k = \sum_{i=1}^{m-1} a_i e_i + \sum_{j=m}^{k} b_j A_2^{j-m} e_m
\end{equation}
for some constants $a_i$ and $b_j$.
This can be shown by induction. 

The base case $k=m$ holds by inspecting the walk matrix of the path $P_m$.

For the inductive step, assume \cref{eqn:induct} holds for up to $k$. We now show it holds for $k+1$.
We have
\begin{align*}
z_{k+1} &= A_X z_k, 
		\ \ \mbox{ by definition of $z_k$ } \\
	&= A_X (\sum_{i=1}^{m-1} a_i e_i + \sum_{j=m}^{k} b_j A_2^{j-m} e_m),
		\ \ \mbox{ by inductive hypothesis } \\
	&= A_1 (\sum_{i=1}^{m-1} a_i e_i) + (A_1 + A_2)b_m e_m + A_2(\sum_{j=m+1}^{k} b_j A_2^{j-m} e_m),
		\ \ \mbox{ since $X$ is a $1$-sum joined at $m$ }
\end{align*}
Since $A_1(\sum_{i=1}^{m-1} a_i e_i) + b_m A_1 e_m$ is in the span of $\{e_1,\ldots,e_m\}$,
we can rewrite the last equation as follows:
\begin{align*}
z_k &= \sum_{i=1}^{m-1} \tilde{a}_i e_i + b_{m-1} e_m + b_m A_2 e_m + \sum_{j=m+1}^{k} b_j A_2^{j-m+1} e_m
\end{align*}
where $b_{m-1} = a_{m-1} e_{m}^T A_1 e_{m-1}$.
After grouping the terms involving the powers of $A_2$ together in the second summand, we obtain
\begin{align*}
z_k &= \sum_{i=1}^{m-1} \tilde{a}_i e_i + \sum_{j=m}^{k+1} \tilde{b}_j A_2^{j-m} e_m
\end{align*}
for some choice of constants $\tilde{b}_j$. This completes the proof of \cref{eqn:induct}.

We are now ready to show that the set $\{z_0,\ldots,z_k\}$ is linearly independent for $k=m,\ldots,m+n-2$.
Assume that $z_{k+1}$ is in the span of $z_0,\ldots,z_k$ for some $k \in \{m,\ldots,m+n-2\}$.
That is, for some constants $\alpha_0,\ldots,\alpha_\ell$, we have
\[
	z_{k+1} = \sum_{\ell=0}^{k} \alpha_{\ell} z_{\ell}.
\]
Now, we apply \cref{eqn:induct} to both sides and obtain the following:
\[
	\sum_{i=1}^{m-1} a_i e_i + \sum_{j=m}^{k+1} b_j A_2^{j-m} e_m
	=
	\sum_{\ell=0}^{k} \alpha_{\ell} (\sum_{i=1}^{m-1} a^{(\ell)}_i e_i + \sum_{j=m}^{\ell} b^{(\ell)}_j A_2^{j-m} e_m).
\]
By projecting both sides to the span of $e_m,\ldots,e_{m+n-2}$, we get
\[
	\sum_{j=m}^{k+1} b_j A_2^{j-m} e_m
	=
	\sum_{\ell=0}^{k} \alpha_{\ell} \sum_{j=m}^{\ell} b^{(\ell)}_j A_2^{j-m} e_m.
\]
After rearranging, this shows that $A_2^{k+1-m} e_m$ is in the span of the vectors $e_m,A_2 e_m,\ldots,A_2^{k-m} e_m$.
But, this contradicts the assumption that the walk matrix of $H$ has full rank.
\end{proof}
\end{lemma}

\begin{figure}[h]
\begin{center}
\begin{tikzpicture}[
    main node/.style={circle,draw,font=\bfseries}, main edge/.style={-,>=stealth'},
    scale=0.5,
    stone/.style={},
    black-stone/.style={black!80},
    black-highlight/.style={outer color=black!80, inner color=black!30},
    black-number/.style={white},
    white-stone/.style={white!70!black},
    white-highlight/.style={outer color=white!70!black, inner color=white},
    white-number/.style={black}]
\tikzset{every loop/.style={thick, min distance=17mm, in=45, out=135}}

% to show particle, uncomment the next line
%\gustone[0]{black}{-3}{1.25}

% ellipse
\draw[fill={gray!20}, drop shadow]
    (-9.0,0.0) ellipse (3.1cm and 1.75cm)
    (+9.0,0.0) ellipse (3.1cm and 1.75cm);

\tikzstyle{every node}=[draw, thick, shape=circle, fill={gray!20}];
\path (-1.5,0.0) node [scale=0.7] (a1) {$a$};
\path (+1.5,0.0) node [scale=0.7] (b1) {$a$};
\path (-6.5,0.0) node [scale=0.7] (aa1) {$b$};
\path (+6.5,0.0) node [scale=0.7] (bb1) {$b$};

\tikzstyle{every node}=[];
\path (-4.0,1.0) node (qleft) {\mbox{$P_m$}};
\path (+4.0,1.0) node (qright) {\mbox{$P_m$}};
\path (-9.0,0.0) node (qqleft) {\mbox{\LARGE $H$}};
\path (+9.0,0.0) node (qqright) {\mbox{\LARGE $H$}};

% middle/bridge edge
\draw[very thick, color=gray]
    (a1) -- (b1);
\draw[thick, color=gray]
    (a1) -- (aa1);
\draw[thick, color=gray]
    (b1) -- (bb1);

\tikzstyle{every node}=[];
\node at (0.0,0.5) {$\alpha$};

\end{tikzpicture}
\caption{The rooted product of $P_2$ with a $1$-sum involving the path $P_m$ and a controllable graph $H$. 
The path $P_m$ can be arbitrarily lengthened (that is $m \rightarrow \infty$) and pretty good state transfer
is preserved by rescaling only the middle edge.
}
\label{fig:chain}
\end{center}
\end{figure}
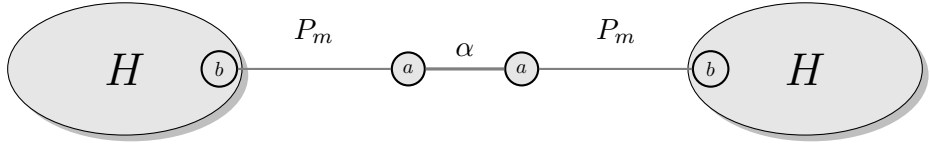

The above corollary allows for an iterative construction of weighted graphs with pretty good state
transfer by ``gluing'' controllable graphs onto weighted paths (see \cref{fig:chain}).

\begin{theorem}
Let $P_m$ be a path on vertices $\{1,\ldots,m\}$ and let $H$ be a graph on vertices $\{m,\ldots,m+n-1\}$ 
that is controllable at vertex $m$.
Let $Y$ be the $1$-sum of $P_m$ and $H$ joined at vertex $m$.
Let $X_\alpha = (\alpha P_2)^{(Y,1)}$ be a rooted product with an involution $\sigma$ that swaps the vertices of $P_2$.
Then, $X_\alpha$ has pretty good state transfer between each pair of vertices $a$ and $\sigma(a)$,
for infinitely many $\alpha \in \RR$.
\end{theorem}

\bigskip
\par\noindent{\em Remarks}.
As a further corollary, we obtain a family with multiple pretty good state transfer between subsets of its vertices. 
This is a loop-free generalization of a similar construction in \cite{aceghtwz25}. We state this without proof as it
follows similarly to \cref{thm:closure}.

\begin{corollary} \label{cor:upst_controllable}
Let $X$ be a graph with universal perfect state transfer (that is, perfect state transfer between every
pair of vertices).
Then, for any controllable graph $H$, the rooted product $(\alpha X)^{H}$ has multiple pretty good state transfer 
within the subset $S_y = \{(x,y) : x \in V(X)\}$ for each vertex $y \in  V(Y)$, for infinitely many $\alpha \in \RR$.

\end{corollary}

\section{Efficient high-fidelity state transfer}

In pretty good state transfer, although we can achieve arbitrarily high fidelity of $1-\epsilon$, 
for any $\epsilon > 0$, the transfer time $\tau$ is non-constructive 
(as it employs Kronecker's approximation theorem). 
But, we show how to circumvent this by exploiting further properties of the rooted product.

For a nonsingular matrix $M$, its {\em condition number} $\kappa(M)$ is the ratio between the largest 
and the smallest eigenvalues of $M$ in absolute values (see \cite{hj13}).
More specifically, let
\[
\lambda_{min}(M) = \min\{|\lambda| : \lambda \in \Sp(M)\},
\ \ \
\lambda_{max}(M) = \max\{|\lambda| : \lambda \in \Sp(M)\}.
\]
Then, the condition number of $M$ is defined as
\[
	\kappa(M) = \frac{\lambda_{max}(M)}{\lambda_{min}(M)}.
\]
So, if $M$ is normalized (or $\lambda_{max}(M)=1$), then $1/\kappa(M)$ 
is distance of the closest eigenvalue of $M$ to zero. 
We say $M$ has good condition number if this distance is large.

We show that the condition number of the rooted product $X^{(Y,r)}$ can largely be
controlled by the condition number of $Y\setminus r$ (modulo a few other assumptions).

\begin{lemma} \label{lemma:controllable-condition}
Let $X$ be a graph and $Y$ be a controllable graph at its vertex $r$.
Assume that $A(Y)$ and $A(Y \setminus r)$ are nonsingular and that $A(X)$ is positive semidefinite with bounded spectral radius.
If 
\begin{equation} \label{eqn:strange}
e_r^T A(Y)^{-1} e_r > 0,
\end{equation} 
then
\begin{equation} \label{eqn:kappa}
	\lambda_{min}(X^{(Y,r)}) \ge \lambda_{min}(Y\setminus r).
\end{equation}

\par\noindent{\em Remark}: The assumption that $A(X)$ is positive definite is not unreasonable.
For example, this can be achieved by taking $(\alpha I + A(X))/3\alpha$, where $\alpha = \lceil 2\norm{A(X)}\rceil$.

\begin{proof}
By \cref{eqn:char-poly} and \cref{lemma:gm-decomposition}, the eigenvalues of the rooted product $X^{(Y,r)}$ are the roots of
the polynomials $\phi_\lambda(t) = \phi(Y,t) - \lambda\phi(Y \setminus r,t)$, as $\lambda \in [0,1]$ ranges over the eigenvalues of $X$. 
If $z_0$ is a zero of $\phi_\lambda(t)$, then $\phi(Y,z_0) = \lambda\phi(Y \setminus r,z_0)$.
Thus, the roots of $\phi_\lambda(t)$ are the points where the two polynomials $\phi(Y,t)$ and $\lambda\phi(Y \setminus r,t)$ intersect.

Note that
\begin{equation} \label{eqn:determinant}
	\frac{\phi(Y\setminus r,t)}{\phi(Y,t)} = e_r^T (tI - A(Y))^{-1} e_r.
\end{equation}
Since $e_r^T A(Y)^{-1}e_r$ is positive, $\phi(Y,0)$ and $\phi(Y\setminus r,0)$ must have opposite signs.

Suppose that $\mu_r$ and $\mu_{r+1}$ are two consecutive eigenvalues of $Y$ closest to zero where $\mu_r < 0 < \mu_{r+1}$. 
Let $I = (\mu_r,\mu_{r+1})$ be the interval around zero formed by these two eigenvalues.
As the polynomials $\phi(Y,t)$ and $\phi(Y \setminus r,t)$ interlace (by Lemma 7.2 in \cite{g12} and Cauchy interlacing), 
there is a unique eigenvalue $\theta$ of $Y \setminus r$ which lies in $I$.

We first consider the case where $\phi(Y,0) > 0$. This implies that $\phi(Y \setminus r,0) < 0$ since $e_r^T A^{-1} e_r > 0$.
If $\theta \in (\mu_r,0)$, then the only intersection of $\phi(Y,t)$ and $\phi(Y \setminus r,t)$ in $I$ 
is in the subinterval $(\mu_r,\theta)$. This argument applies to $\lambda\phi(Y\setminus r,t)$ as $\lambda \in [0,1]$, and
hence it proves \cref{eqn:kappa}.
On the other hand, if $\theta \in (0,\mu_{r+1})$, then the only intersection of $\phi(Y,t)$ and $\phi(Y \setminus r,t)$ 
in $I$ lies in the subinterval $(\theta,\mu_{r+1})$. This again proves \cref{eqn:kappa} by a similar reasoning.

The case where $\phi(Y,0) < 0$ (and hence $\phi(Y \setminus r,0) > 0$) is similar.
\end{proof}
\end{lemma}

We state the next result which will be used in our subsequent theorem.
The result states that high-fidelity state transfer can be achieved by attaching weak pendant edges on
two cospectral vertices that belongs to a sufficiently nonsingular graph.

\begin{theorem} \label{thm:trex} (Kay and Tamon \cite{kt}, see Theorem 4.1) \\
Let $X_0$ be a connected graph whose adjacency matrix satisfies $\norm{A(X_0)}=1$ with a finite condition number $\kappa$.
Let $X = X_0 \sqcup \overline{K}_2$ be a disjoint union of $X_0$ and two isolated vertices. %with adjacency matrix $H$.
Suppose $y_a,y_b$ are cospectral vertices in $X_0$ and $\delta \in (0,1)$ satisfy
$\delta\kappa \ll 1$ and $|\lip{e_{y_b}}{A(X_0)^{-1}e_{y_a}}| \gg \delta^2\kappa^3$.
If $z_a,z_b$ are the vertices of $\overline{K}_2$ and $W$ is a matrix which represents the edges 
$(y_a,z_a)$ and $(y_b,z_b)$, then 
\begin{equation}
    |\lip{e_{z_b}}{e^{-i\tau(A(X) + \delta W)} e_{z_a}}| \ge 1-\OO(\delta)
\end{equation}
with transfer time $\tau = \frac{\pi}{2} \frac{1}{\delta^2|\lip{e_{y_b}}{A(X_0)^{-1}e_{y_a}}|}$.
\end{theorem}

We collect some basic observations on the spectrum of a rooted product. 
Note that $\lambda_{max}(X) = \norm{A(X)}$.

\begin{claim} \label{claim:bounded-norm}
Let $X$ be a graph and let $Y$ be a graph controllable at $r$. Then,
\[
	\norm{A(X^Y)} \le \norm{A(X)} + \norm{A(Y)}.
\]

\begin{proof}
Using $A(X^Y) = A(Y) \otimes I + e_r e_r^T \otimes A(X)$, we have
\begin{align*}
\norm{A(X^Y)} &\le \norm{A(Y) \otimes I} + \norm{e_r e_r^T \otimes A(X)}, 
		\ \ \mbox{ by subadditivity } \\
	&\le \norm{A(Y)} + \norm{e_r e_r^T}\norm{A(X)}, 
		\ \ \mbox{ by submultiplicativity } \\
	&\le \norm{A(Y)} + \norm{A(X)}.
\end{align*}
We had used the fact that the norm of $e_r e_r^T$ is at most $1$ (since it is a projector).
\end{proof}
\end{claim}

\begin{claim} \label{claim:matrix-inverse}
Let $X$ be a graph with nonsingular adjacency matrix. Let $a,b$ be two strongly cospectral vertices in $X$. Then,
\[
	\frac{1}{\lambda_{max}(X)} \le |e_b^T A(X)^{-1} e_a| \le \frac{1}{\lambda_{min}(X)}.
\]

\begin{proof}
If $A(X) = \sum_r \theta_r E_r$ is the spectral decomposition of $A(X)$, then
\begin{align*}
e_b^T A(X)^{-1}e_a 
	&= \sum_r \frac{1}{\theta_r} e_b^T E_r e_a \\
	&= \sum_r \frac{1}{\theta_r} \sqrt{e_b^T E_r e_b} \sqrt{e_a^T E_r e_a},
		\ \ \mbox{ by Cauchy-Schwarz and strong cospectrality} \\
	&= \sum_r \frac{1}{\theta_r} e_a^T E_r e_a, 
		\ \ \mbox{ by cospectrality $e_a^T E_r e_a = e_b^T E_r e_b$}.
\end{align*}
Note that the Cauchy-Schwarz inequality is tight due to strong cospectrality ($E_r a_a$ is parallel to $E_r e_b$).
The claim now follows by bounding $1/\theta_r$ by $1/\lambda_{min}$ from above and by $1/\lambda_{max}$ from below.
\end{proof}
\end{claim}

Now, we are ready to state and prove our main observation on high-fidelity state transfer on the rooted
product of a graph with a controllable graph that has a good condition number (see \cref{fig:tulip}).

\begin{theorem}
Let $X$ be a graph where $A(X)$ is positive definite with bounded norm.
Let $Y$ be a graph controllable at $r$ where both $A(Y)$ and $A(Y \setminus r)$ are nonsingular and
have bounded norms. Assume that $\lambda_{min}(Y \setminus r) = \Omega(1)$, and $e_r^T A(Y)^{-1}e_r > 0$.
For two {\em strongly} cospectral vertices $a,b \in V(X)$ and a vertex $y \in V(Y)$,
let $Z_{a,b}^y$ be a graph obtained from $X^{(Y,r)}$ by adding two new vertices $z_a,z_b$ and 
two pendant edges $(y_a,z_a)$ and $(y_b,z_b)$ of weight $\delta \ll 1$.
Then, quantum state transfer between $z_a$ and $z_b$ occurs in $Z_{a,b}^y$ with fidelity $1-\OO(\delta)$ 
at time $\tau = \frac{\pi}{2} \frac{1}{\delta^2\Delta}$ where $\Delta \asymp 1$.

\begin{proof}
Let $\kappa = \lambda_{min}(Y \setminus r)$ where $\kappa = \Omega(1)$.
By \cref{lemma:controllable-condition}, we have that $\lambda_{min}(X^{(Y,r)}) \ge \kappa$. 
To apply \cref{thm:trex}, we require $\delta\kappa \ll 1$. Since $\kappa = \Omega(1)$,
it suffices to choose $\delta \ll 1$. For \cref{thm:trex}, we also require
$\Delta := |e_{y_b}^T A(X^{Y})^{-1} e_{y_a}|$ to satisfy $\Delta \gg \delta^2\kappa^3$.
Again, as $\kappa = \Omega(1)$, it suffices to ensure $\Delta \gg \delta^2$.
By \cref{claim:matrix-inverse}, we know that 
\[
	\frac{1}{\lambda_{max}(X^Y)} \ \le \ \Delta \ \le \ \frac{1}{\lambda_{min}(X^Y)}.
\]
By \cref{claim:bounded-norm}, $A(X^Y)$ has bounded norm since both $X$ and $Y$ have bounded norms.
This implies $\lambda_{max}(X^Y) = \OO(1)$. Therefore, $\Delta \asymp 1$, which clearly satisfies
$\Delta \gg \delta^2$ since $\delta \ll 1$. 
\end{proof}
\end{theorem}

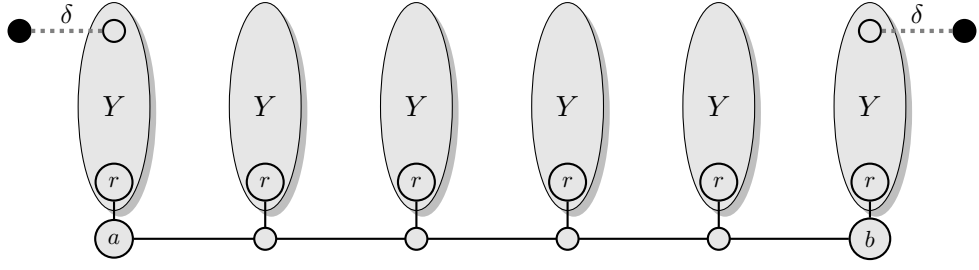
\begin{figure}[t]
\begin{center}
\begin{tikzpicture}[
% T.rex
    main node/.style={circle,draw,font=\bfseries}, main edge/.style={-,>=stealth'},
    scale=0.5,
    stone/.style={},
    black-stone/.style={black!80},
    black-highlight/.style={outer color=black!80, inner color=black!30},
    black-number/.style={white},
    white-stone/.style={white!70!black},
    white-highlight/.style={outer color=white!70!black, inner color=white},
    white-number/.style={black}]
\tikzset{every loop/.style={thick, min distance=17mm, in=45, out=135}}

% to show particle, uncomment the next line
%\gustone[0]{black}{-3}{1.25}

% ellipse
\draw[fill={gray!20}, drop shadow]
    (-10.0,3.5) ellipse (0.95cm and 2.75cm)
    (-6.0,3.5) ellipse (0.95cm and 2.75cm)
    (-2.0,3.5) ellipse (0.95cm and 2.75cm)
    (+2.0,3.5) ellipse (0.95cm and 2.75cm)
    (+6.0,3.5) ellipse (0.95cm and 2.75cm)
    (+10.0,3.5) ellipse (0.95cm and 2.75cm);

\tikzstyle{every node}=[draw, thick, shape=circle, fill={gray!20}];
\path (-10.0,0.0) node [scale=0.8] (a1) {$a$};
\path (+10.0,0.0) node [scale=0.8] (a6) {$b$};

% base path
\draw[thick]
    (a1) -- (a6);

\path (-10.0,1.5) node [scale=0.8] (aa1) {$r$};
\path (-10.0,5.5) node [scale=0.8] (b1) {};
\path (-6.0,0.0) node [scale=0.8] (a2) {};
\path (-6.0,1.5) node [scale=0.8] (aa2) {$r$};
\path (-2.0,0.0) node [scale=0.8] (a3) {};
\path (-2.0,1.5) node [scale=0.8] (aa3) {$r$};
\path (+2.0,0.0) node [scale=0.8] (a4) {};
\path (+2.0,1.5) node [scale=0.8] (aa4) {$r$};
\path (+6.0,0.0) node [scale=0.8] (a5) {};
\path (+6.0,1.5) node [scale=0.8] (aa5) {$r$};
\path (+10.0,1.5) node [scale=0.8] (aa6) {$r$};
\path (+10.0,5.5) node [scale=0.8] (b6) {};

\tikzstyle{every node}=[draw, thick, shape=circle, fill={black}];
\path (-12.5,5.5) node [scale=0.8] (c1) {};
\path (+12.5,5.5) node [scale=0.8] (c6) {};

\draw[thick]
	(a1) -- (aa1)
	(a2) -- (aa2)
	(a3) -- (aa3)
	(a4) -- (aa4)
	(a5) -- (aa5)
	(a6) -- (aa6);

% weak couplings
\draw[dotted, ultra thick, color=gray]
    (b1) -- (c1)
    (b6) -- (c6);

\tikzstyle{every node}=[];
\node at (-11.25,5.95) {\mbox{\small $\delta$}};
\node at (+11.25,5.95) {\mbox{\small $\delta$}};

\tikzstyle{every node}=[];
\path (-10.0,3.5) node (q1) {$Y$};
\path (-6.0,3.5) node (q1) {$Y$};
\path (-2.0,3.5) node (q1) {$Y$};
\path (+2.0,3.5) node (q1) {$Y$};
\path (+6.0,3.5) node (q1) {$Y$};
\path (+10.0,3.5) node (q1) {$Y$};

\end{tikzpicture}
\caption{Efficient high-fidelity state transfer between the dark vertices in time $1/\delta^2$ 
tunneling through the rooted product $X^Y$ where $Y$ is controllable at $r$.
This is facilitated by strong cospectrality of $a$ and $b$ in the base graph $X$
(no state transfer is required between $a$ and $b$).
Here, as an example, $X$ is a path.
}
\label{fig:tulip}
\end{center}
\end{figure}

%%%%%%%%%%%%%%%%%%%%%%%%%%%%%%%%%%%%%%%%%%%%%%%%%%%%%%%%%%%%%%%%%%%%%%%%%%%%%%%%%%%%%%%%%%%%%%%%%%%%%
\section{Concluding remarks}

In this work, we showed that a classical tool in algebraic graph theory called rooted product is useful
for quantum state transfer. 

First, we proved a simple transference principle for rooted product of a base graph with state transfer 
and a controllable graph. Here, the state transfer properties of the base graph will be inherited by 
the rooted product. Second, we showed that, even in the absence of state transfer in the base graph, 
the rooted product is useful in constructing graphs with efficient high-fidelity state transfer. 
For the latter, we exploited the fact that the condition number of the rooted product is controlled 
by the pendant (controllable) subgraph. 

Our second contribution also circumvented the use of transcendental weights for pretty good state transfer.
In most realistic scenarios, it will be hard to prepare weights that are transcendental (due to precision
and rounding errors). It is true that one may employ Hilbert Irreducibility Theorem (HIT) to replace 
transcendental numbers, but the complexity issues surrounding effective versions of Hilbert's theorem
remains nontrivial to the best of our knowledge.

We conclude with the following open questions for future explorations:
\begin{enumerate}
\item Is it possible to extend the results to a rooted product between a base graph $X$ of order $n$
	and a sequence of distinct rooted graphs $\{(Y_1,r_1),\ldots,(Y_n,r_n)\}$?
	Makmal \etal \cite{mzmtb14} had studied the case when $X$ is a $n$-cube and $Y_\ell$ are paths
	with different lengths.

\item Given that random graphs are controllable almost surely (see \cite{ot16}), our construction 
	works on almost all graphs. Are there natural classes of controllable graphs useful for 
	quantum state transfer?

\end{enumerate}

%%%%%%%%%%%%%%%%%%%%%%%%%%%%%%%%%%%%%%%%%%%%%%%%%%%%%%%%%%%%%%%%%%%%%%%%%%%%%%%%%%%%%%%%%%%%%%%%%%%%%
\section*{Acknowledgments}

This work is supported by NSF ExpandQISE grant 2427020.
We thank Ada Chan for helpful discussions.
All results presented in this work were discovered, proved and written solely by the human authors.

%%%%%%%%%%%%%%%%%%%%%%%%%%%%%%%%%%%%%%%%%%%%%%%%%%%%%%%%%%%%%%%%%%%%%%%%%%%%%%%%%%%%%%%%%%%%%%%%%%%%%

%%%%%%%%%%%%%%%%%%%%%%%%%%%%%%%%%%%%%%%%%%%%%%%%%%%%%%%%%%%%%%%%%%%%%%%%%%%%%%%%%%%%%%%%%%%%%%%%%%%%%
\end{document}